\documentclass[a4paper,USenglish,cleveref,autoref, thm-restate,numberwithinsect]{lipics-v2021}

\pdfoutput=1

\hideLIPIcs  

\title{Testing \texorpdfstring{$H$}{H}-freeness on sparse graphs, the case of bounded expansion}

\author{Samuel Humeau}{ENS de Lyon, CNRS, LIP, UMR 5668, 69342, Lyon cedex 07, France}{samuel.humeau@ens-lyon.fr}{https://orcid.org/0009-0007-1850-9744}{}
\author{Mamadou Moustapha Kanté}{Université Clermont Auvergne, Clermont Auvergne INP, LIMOS, CNRS, Clermont-Ferrand, France}{mamadou.kante@uca.fr}{https://orcid.org/0000-0003-1838-7744}{}
\author{Daniel Mock}{RWTH Aachen University, Germany}{mock@cs.rwth-aachen.de}{https://orcid.org/0000-0002-0011-6754}{Supported by the Deutsche Forschungsgemeinschaft
(DFG, German Science Foundation) – DFG-927/15-2}
\author{Timothé Picavet}{LaBRI, Université de Bordeaux, France}{timothe.picavet@u-bordeaux.fr}{https://orcid.org/0000-0002-7129-0127}{}
\author{Alexandre Vigny}{University Clermont Auvergne, France}{alexandre.vigny@uca.fr}{https://orcid.org/0000-0002-4298-8876}{}

\authorrunning{S. Humeau, M.M. Kanté, D. Mock, T. Picavet, and A. Vigny} 

\Copyright{Samuel Humeau, Mamadou Moustapha Kanté, Daniel Mock, Timothé Picavet, and Alexandre Vigny} 

\ccsdesc[500]{Mathematics of computing~Graph algorithms}
\ccsdesc[300]{Theory of computation~Streaming, sublinear and near linear time algorithms}

\keywords{Property testing, Sparsity, Bounded expansion, Treedepth} 

\category{} 

\relatedversion{} 

\nolinenumbers 

\EventEditors{John Doe}
\EventNoEds{2}
\EventLongTitle{42nd Conference on Very Important Topics (CVIT 2016)}
\EventShortTitle{CVIT 2016}
\EventAcronym{CVIT}
\EventYear{2016}
\EventDate{December 24--27, 2016}
\EventLocation{Little Whinging, United Kingdom}
\EventLogo{}
\SeriesVolume{42}
\ArticleNo{23}

\usepackage[ruled,vlined,linesnumbered]{algorithm2e}

\usepackage{algorithmic}
\usepackage{forest}
\usepackage{amsfonts}\usepackage{stmaryrd,amssymb,amsmath,amsthm,amsfonts}
\usepackage{xspace,xpunctuate}
\usepackage{mathtools}
\usepackage{graphicx}
\usepackage{pgf,tikz}
\usetikzlibrary{cd}
\usepackage[normalem]{ulem} 
\usepackage{listings}

\renewcommand{\H}{\ensuremath{{\mathcal{H}}}\xspace}

\newcommand{\Cc}{\ensuremath{{\mathcal{C}}}\xspace}
\newcommand{\Hh}{\ensuremath{{\mathcal{H}}}\xspace}
\newcommand{\Ff}{\ensuremath{{\mathcal{F}}}\xspace}
\newcommand{\e}{\ensuremath{\epsilon}}
\newcommand{\bfstraverse}{{\normalfont{\texttt{Random\-Bounded\-BFS}}}\xspace}
\newcommand\td{\operatorname{td}}

\begin{document}
\maketitle

\begin{abstract}
	In property testing, a tester makes queries to (an oracle for) a graph and, on a graph having or being far from having a property $P$, it decides with high probability whether the graph satisfies $P$ or not. Often, testers are restricted to a constant number of queries. While the graph properties for which there exists such a tester are somewhat well characterized in the dense graph model, it is not the case for sparse graphs.
	In this area, Czumaj and Sohler (FOCS’19) proved that~$H$-freeness (i.e.~the property of excluding the graph~$H$ as a subgraph) can be tested with constant queries on planar graphs as well as on graph classes excluding a minor.

	Using results from the sparsity toolkit, we propose a simpler alternative to the proof of Czumaj and Sohler, for a statement generalized to the broader notion of bounded expansion.
	That is, we prove that for any class $\Cc$ with bounded expansion and any graph~$H$, testing $H$-freeness can be done with constant query complexity on any graph $G$ in $\Cc$, where the constant depends on $H$ and $\Cc$, but is independent of $G$.

	While classes excluding a minor are prime examples of classes with bounded expansion, so are, for example, cubic graphs, graph classes with bounded maximum degree, graphs of bounded book thickness, or random graphs of bounded average degree.
\end{abstract}


\section{Introduction}

Given some data with a question, do you need to read all the data to answer, or can you parse only a fraction of it before answering correctly with high probability?
The domain of study associated to such problems is known as \emph{property testing}.
Intuitively, a property testing algorithm (or {\em tester}) takes a graph~$G$ and a property~$P$ as inputs, accepts with high probability if~$G$ satisfies~$P$, and rejects with high probability if~$G$ is ``\emph{far}'' from having property~$P$.
The notion of being ``far'' from having a property is often described via the fraction $\e$ of edge additions or edge removals needed to transform~$G$ into a graph~$G'$ satisfying~$P$. The fixed fraction $\e$ is often called the \emph{proximity parameter}. A tester accepting with probability one on graphs satisfying $P$ has \emph{one-sided error}. Otherwise, the tester has \emph{two-sided error}.

A tester does not have direct access to its input graph. Instead it uses an \emph{oracle}, i.e. a ``black box'' serving as a representation of the input graph to which \emph{queries} can be made. These depend on context and can be, for example: ``are these two vertices neighbors?'' or ``give me a random neighbor of this vertex''.

There exists several models formalizing the notions described above: the dense, sparse, and bounded degree models, and they differ on how ``being far from a property'' is defined and on the available queries.
The efficiency of testers is quantified by the number of queries, called \emph{query complexity}, made to a given representation of $G$. Ideally, the number of queries is constant: it depends on the property $P$ and the proximity parameter $\e$ of ``being far from $P$'', but not on $G$. In this case, $P$ is called {\em testable with constant query complexity}, or \emph{testable}.

More formally, in the {\em dense graph model} introduced by Goldreich, Goldwasser, and Ron~\cite{GoldreichGR98}, a graph $G$ is said to be \emph{$\e$-close} to a property $P$ if changing at most a fraction $\e$ of its adjacency matrix (i.e. at most~$\e\cdot|V(G)|^2$ many edges) transforms $G$ into a graph $G'$ satisfying~$P$. Otherwise $G$ is called \emph{$\e$-far} from $P$.
For this model, there is a good understanding of what properties are testable with two-sided error~\cite{AlonFNS09}, but also one-sided~\cite{AlonS08,AlonS08a}.
For example, every hereditary property (i.e. stable under taking induced subgraphs; such as bipartite, $H$-free or $k$-colorable) are one-sided error testable in the dense model.

In the bounded degree model \cite{GoldreichR02}, the input comes with an integer $d$ and every vertex of the input is assumed to have degree at most $d$.
In this case, it is known that restricting the class further to planar graphs, or graphs that exclude a minor (i.e. proper minor closed classes of graphs), and, more generally, to hyperfinite graph classes, allows every property to be testable with two-sided error~\cite{NewmanS13}.
For the more general case where only the degree constraint is considered, far fewer properties can be tested efficiently. On the one hand, it is known that testing FO-properties of the form $\exists\forall$ can be done efficiently. This includes testing for $H$-freeness. On the other hand some $\forall\exists$-FO formulas cannot be tested efficiently~\cite{AdlerKP24}. One of the most important property that cannot be tested is bipartiteness, which requires (on general graphs with bounded degree) $\Omega(\sqrt{|G|})$ many queries. See 
chapter nine of~\cite{Goldreich17}.

The sparse model is currently the less known model. It splits into slightly different variations, depending on whether from a given vertex the algorithm can query: 1) it's $i$th neighbor (upon input an integer $i$, and receives an error if the vertex has less than $i$ neighbors), or 2) a random neighbor (possibly several times the same), or 3) a random distinct neighbor (distinct from the outputs of previous queries on the same vertex). 
We refer to the introduction of \cite{CzumajS19} for more information and, similarly to the recent work in this
area~\cite{AwofesoGLLR25,AwofesoGLR25,CzumajMOS19,EsperetN22}. We focus here on the \emph{random neighbor model}, i.e., the second variation of the sparse model as presented above.
Observe that a query in this variation can be simulated using a constant number of queries in the other variations. Hence any property testable in the random
neighbor model is testable in every variation of the sparse model. In a similar manner, a query from any variation of the sparse model can be simulated using a
constant number of queries in the bounded degree model \cite[Section~1.1.2]{CzumajS19}. Hence, it is natural to seek generalization of results from the sparse and bounded degree models, starting with the random neighbor model. Formal definitions for the latter are recalled in the preliminaries.

\medskip\noindent{\bf State of the art.}
The most impactful work in the sparse model is the tester for subgraph freeness on planar graphs (and actually all classes excluding a fixed minor) of Czumaj and Sohler~\cite{CzumajS19}. Their algorithm simply repeats random breadth first searches with breadth and depth
independent of the input graph.
There are several possible directions to generalize this work. First, providing algorithms testing broader properties. This is the case of~\cite{EsperetN22}, which uses the algorithm of \cite{CzumajS19} as a subroutine and shows how to test any monotone property in the sparse model.

Another direction is to look at more general graph classes, extending beyond minor-free graph classes. Minor-free graph classes do not contain all cubic graphs, so that the results in the bounded degree model are separated from the ones of~\cite{CzumajS19}.
So far, such an extension has only been achieved for testing the non-existence of cycles~\cite{AwofesoGLLR25,AwofesoGLR25} (on classes with bounded $r$-admissibility), and as Esperet and Norin noted in \cite{EsperetN22}: {\it ``However, since the proof of~\cite{CzumajS19} itself strongly relies on edge-contractions (and thus on the graph class $\Cc$ being minor-closed), Theorem 2 [of~\cite{EsperetN22}] does not seem to be easily extendable beyond minor-closed classes.''}.

Furthermore, the proof of \cite{CzumajS19} is quite long and technical with key lemmas being dependent on their technical machinery. It is therefore hard to use corresponding insightful ideas individually. For examples, \cite{EsperetN22} uses Czumaj and Sohler's tester as a black box, while \cite{AwofesoGLLR25,AwofesoGLR25} do not seem to use any of their tools.

\medskip\noindent{\bf Our contribution.}
We prove that, for every proximity parameter, graph class $\cal C$ of bounded expansion, and property $P$ characterized by a finite set of forbidden subgraphs, $P$ is testable on graphs of $\cal C$ in the sparse model with constant query complexity and one-sided error.
This result generalizes that of ~\cite{CzumajS19} from proper minor-closed graph classes to graph classes $\cal C$ of bounded expansion.
In particular, we get that such properties $P$ are testable on graph classes of bounded degree and proper topologically minor-closed classes, generalizing results from the bounded degree model as well. A part of the hierarchy of sparse graph classes including these is represented in \Cref{pic:graphclasses}.

We make sure to keep lemmas as independent of each other as possible, clearly stating what the requirements of each statement are. While some lemmas work on graphs from a bounded expansion class, some only require their degeneracy to be bounded, and some hold for arbitrary graphs. Lemmas and definitions which are reformulations or slight variations of statements of \cite{CzumajS19} are explicitly pointed out as such.

\begin{figure}[t]
  \centering
  \includegraphics[width=0.9\linewidth]{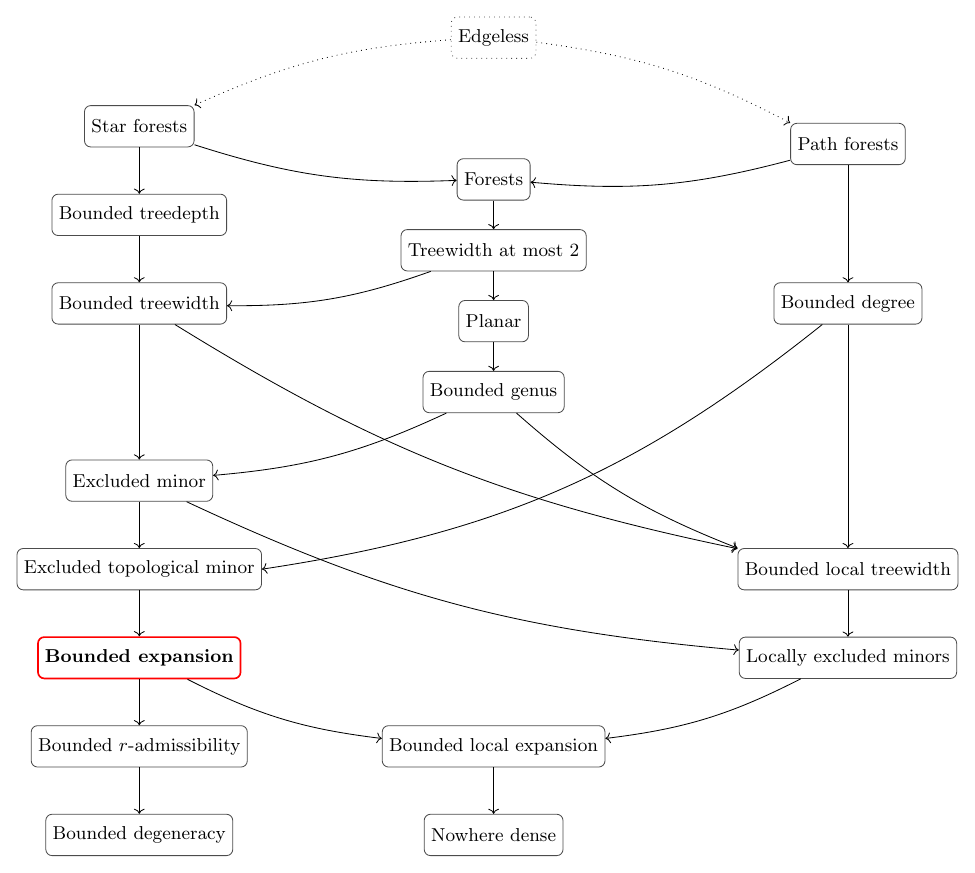}
  \caption{Families of sparse graph classes and their pairwise inclusions.}\label{pic:graphclasses}
\end{figure}

\smallskip\noindent{\bf Organization and techniques.}
To prove our result, we first proceed in the preliminaries in the same manner as Czumaj and Sohler \cite{CzumajS19}: we focus on the testability of $H$-freeness for a fixed graph, and we use the same tester: repeated bounded random breadth first searches accepting if a copy of $H$ is found by the searches, and rejecting otherwise. One-sidedness is direct, so most of the proof is to handle graphs which are $\e$-far from being $H$-free. As pointed out by Czumaj and Sohler, such a graph $G$ contain $\Omega(\e|V(G)|)$ edge-disjoint copies of $H$.

Hence, during our proof procedure, we handle and maintain a ``sufficiently large'' set $\mathcal H$ of edge-disjoint copies of $H$ in $G$.
With each lemma, we restrict $\mathcal H$ to a subset of itself, imposing further structure on the subgraph $G[\mathcal H]$ induced by the copies of $H$ in $\mathcal H$.
In \cref{sec:reduction_subgraph}, we study a property on $\Hh$ enabling to transfer the probability of finding $H$ in $G[\Hh]$ to that of finding $H$ in $G$.
Later in \cref{sec:from-BE-to-TD}, we explain how to build a set $\Hh$ such that $G[\Hh]$ has bounded treedepth.
Using that having bounded treedepth implies excluding a minor, together with the result of~\cite{CzumajS19}, we prove our main contribution in \cref{sec:ccl-reduction}: \cref{lem:testing_H_freeness}.

For the sake of completeness, and in order not to rely on~\cite{CzumajS19}, we show how to test for $H$-freeness on graphs with bounded treedepth in \cref{sec:bounded_treedepth}. This is done by first analyzing the structure of $H$ with regard to tree embeddings of bounded depth, and later, by analyzing how parts of edge-disjoint copies of $H$ can be used to reconstruct a copy of $H$ in a graph of bounded treedepth. With \cref{lem:testing_H_freeness_treedepth}, we finally conclude that $H$-freeness can be tested with constant query complexity and one-sided error in the random neighbor model on graphs of bounded treedepth. 


\section{Preliminaries}

\subsection{Graph notations}
We denote the vertex and edge set of a graph $G$ by $V(G)$ and $E(G)$, respectively. We write~$|G|$ for the number of vertices of $G$.
We assume that graphs are finite, simple and undirected.

Given a graph $G$ and a set of vertices $A$ of $G$, we denote by $G[A]$ the {\em subgraph} of $G$ induced by the vertices in $A$ (i.e.~keeping every edge
incident to two vertices of $A$). For a set $F$ of edges of a graph $G$, we denote by $G-F$ the subgraph of $G$ with vertex set $V(G)$ and edge set
$E(G)\setminus F$. 

The {\em degree} $\deg^{G}(v)$ of a vertex $v$ in a graph $G$ is the number of vertices $v$ shares an edge with in $G$. The {\em degree} $\Delta(G)$ of a graph $G$ is the maximum degree of its vertices.

\begin{definition}
  A \emph{copy} $h$ of a graph $H$ in a graph $G$ is a subgraph isomorphism, i.e. an injection $h \colon V(H) \to V(G)$ such that $(u,v) \in E(H) \implies (h(u), h(v)) \in E(G)$ for every~$u,v\in V(H)$.

  Given a set $\Hh$ of copies of $H$ in $G$, $G[\mathcal H]$ denotes the subgraph of $G$ defined by the union of the subgraphs of $G$
  image of elements of $\Hh$.
\end{definition}

A graph $G$ is said to be \emph{$H$-free} if it does not contain $H$ as a subgraph, i.e.,~there are no copies of~$H$ in $G$. Note that replacing subgraph by induced subgraph is of little impact as an induced copy of $H$ is a copy of $H$, and a copy of $H$ is and induced copy in $G'$ after removing at most $|E(H)|$-many edges from $G$. 

\subsection{Sparsity}

We start with the definition of treedepth, and then move on to the more general notion of bounded expansion introduced by Nešetřil and Ossona de Mendez. We refer to their more comprehensive book on the subject for more information \cite{sparsity}. 

A {\em tree order} $\le_T$ on a vertex set $V$ is a partial order such that for every vertex $v$, the set of elements $\le_T$-smaller than $v$ is totally ordered. A \emph{root} with regard to $\le_T$ is a $\le_T$-minimal element, and a \emph{leaf} with regard to $\le_T$ is a $\le_T$-maximal element.
The {\em level} of a vertex $v$ in a tree order is the number of elements smaller than $v$. The {\em depth} of a tree order is the maximal level of its vertices.
In particular, roots have level zero.
A tree order $\le_T$ on the vertex set of a graph is called a {\em tree embedding} (or treedepth embedding) if every pair of adjacent vertices $(u,v)$ is $\le_T$-comparable (i.e.~either $u\le_T v$, or $v\le_T u$).

\begin{definition}
  The {\em treedepth} $\td(G)$ of a graph $G$ is the minimum depth of its tree embeddings.

  For 
  an
  integer $d$, a graph class $\Cc$ is said to have {\em treedepth} at most $d$ if the treedepth of each of its graphs is bounded by $d$;
$\Cc$ has \emph{bounded treedepth} if there exists such a constant $d$.
\end{definition}

The notion of bounded expansion was introduced by Ne\v{s}et\v{r}il and Ossona de Mendez~\cite{NesetrilM08}. It has many equivalent characterizations.
As our proofs rely only on the existence of low-treedepth colorings, we only recall the corresponding characterization.

A \emph{$p$-treedepth coloring} of a graph $G$ is a coloring of the vertices of $G$ such that for every set $S$ of $i \leq p$ colors, the subgraph of $G$ induced by its vertices colored by elements of $S$ has treedepth at most $i$.

\begin{definition}[See Theorem 7.1 of \cite{NesetrilM08}]
  \label{def:bounded_expansion}
  A class of graphs $\cal C$ has {\em bounded expansion} if and only if there exists a function $f$ such that for every graph $G \in \cal C$ and all $p \in \mathbb N$, $G$ admits a $p$-treedepth coloring using at most $f(p)$ colors.
\end{definition}

Note that, while such colorings can be computed efficiently, our proof only uses their existence. There are many variations on such coloring, often grouped in
the notion of {\em generalized coloring numbers}. This includes among other weak-coloring, $p$-centered coloring, and admissibility coloring. The latter is used in \cite{AwofesoGLLR25,AwofesoGLR25} as ($r$-) admissibility numbers. In our case, we sometimes use the notion of degeneracy, which is equivalent to the notion of~1-admissibility.

\begin{definition}
  A graph has \emph{degeneracy} $d$ if every subgraph has a vertex of degree at most~$d$.
  A class of graphs has \emph{bounded degeneracy} if there exists a constant $d$ such that every graph in the class has degeneracy at most $d$.
\end{definition}

Classes of bounded expansion are well-known to have bounded degeneracy~\cite[Fact 3.2]{NesetrilM08}.

\subsection{Property testing and complexity}

For a positive constant $\e>0$, a graph $G$ is {\em $\e$-far from a property $P$} if more than $\e\cdot|G|$ many edges have to be deleted from or inserted to $G$ to obtain a graph with property $P$. The parameter $\e$ is called the \emph{proximity parameter}.

A {\em tester} (for a property $P$) is an algorithm that inspects only a part of an input graph~$G$ via a restricted set of \emph{queries} to an \emph{oracle}, accepts with probability at least $2/3$ if $G$ satisfies~$P$, and rejects with probability at least $2/3$ if $G$ is~$\e$-far from~$P$.
A tester has {\em one-sided error} if it accepts all graphs with property $P$ with probability $1$.

In \cite{CzumajS19}, several oracle access models are presented, but the focus is put on the random neighbor model. We proceed similarly.

Initially, a tester knows the vertices of its input graph $G$; it accesses the input graph by making {\em random neighbor queries} to its oracle: given a vertex $v \in V(G)$, the oracle returns a vertex chosen independently and uniformly at random from the set of all neighbors of $v$. No other queries are available to testers.

The {\em query complexity} of a tester is the number of queries it makes before reaching a conclusion.

In our work, we only consider \emph{constant query complexities}: the number of queries is always independent of the input graph.
More precisely, when testing for $H$-freeness with proximity parameter $\e$, the number $q$ of queries can depend on $H$, $\e$, and additional parameters (such as a fixed class $\Cc$ the graph $G$ belongs to). Hence, $q=O_{H,\e,\Cc}(1)$, that is, for some function~$f(H,\e,\Cc)$, we have $q\le f(H,\e,\Cc)$. This notation extends easily, by ignoring multiplicative factors depending on indices; for example stating $|\mathcal H|=\Omega_{\e,H}(|V(G)|)$
for a set of copies of~$H$ in $G$.

\subsection{Random Bounded BFS and the tester}
First we show that testing for $\mathbb H$-freeness with $\mathbb H$ a
set of graphs is equivalent to testing for $H$-freeness for a single
graph $H$. This is identical to a part of \cite[Section 8]{CzumajS19}.
We only prove it for one-sided error testers since this is what we later investigate. 

First, we relate $\e$-farness from being
$\mathbb H$-free or $H$-free for $H$ some graph in
$\mathbb H$:
\begin{proposition}
  Let $G$ be a graph and $\mathbb H$ a finite set of graphs.

  If $G$ is $\e$-far from being $\mathbb H$-free, then there
  exists a graph $H\in\mathbb H$ such that $G$ is $\e/|\mathbb H|$-far
  from being $H$-free.
\end{proposition}
\begin{proof}
  For the sake of contradiction, assume the statement is false. Hence,
  for each graph $H$ of $\mathbb H$, it is possible to remove at most
  $\e|G|/|\mathbb H|$ edges from $G$ to make it $H$-free. If we do
  this for all graphs of $\mathbb H$, then we remove at most $\e|G|$
  edges from $G$ and make it $\mathbb H$-free, a contradiction.
\end{proof}

\begin{proposition}\label{prop:H-Hh}
  Let $\mathbb H$ be a finite set of graphs and $\mathcal C$ a graph class.

  If, for every graph $H$ of $\mathbb H$ and $\e>0$,
  $H$-freeness is testable with constant query complexity,
  proximity parameter $\e$, and one-sided error in the random neighbor
  model on graphs of $\cal C$, then so is $\mathbb H$-freeness.
\end{proposition}
\begin{proof}
  Given $G,\e,\mathbb H$, consider the algorithm iteratively testing $H$-freeness (with proximity parameter~$\e/|\mathbb{H}|$) for each $H\in\mathbb H$. If the sub-tester for one of the graph $H\in\mathbb H$ rejects, then the algorithm rejects $G$. Otherwise, the algorithm accepts $G$. The number of queries is constant (depending only on $\e$ and~$\mathbb H$, and possibly $\mathcal C$).

  Note that if $G$ is indeed $\mathbb H$-free, since sub-testers are one-sided, they all accept, and therefore our algorithm accepts. Otherwise, if $G$ is $\e$-far from $\mathbb H$-freeness, by \cref{prop:H-Hh}, there is exists $H\in\mathbb H$ such that $G$ is $\e/|\mathbb H|$-far from being $H$-free. 
  With probability at least~$2/3$ the sub-tester for $H$ rejects, and therefore our algorithm rejects with probability at least~$2/3$ too.
\end{proof}

Our tester for $H$-freeness is relatively straightforward. It is the same as that of \cite{CzumajS19}, except that we provide formal details on how disconnected graphs are handled, whereas this particular case is only discussed informally in \cite[Section 8]{CzumajS19}. Our tester is built on the subroutine for Random Bounded Breadth First Searches: \bfstraverse described below.

\begin{algorithm}[H]
\caption{$\bfstraverse(G,t,d)$}
\KwIn{Graph $G$, integers $t,d$}

Pick a vertex $v$ uniformly at random in $G$\;
$\texttt{current\_vertices} \gets \{v\}; \quad \texttt{final\_edges} \gets \emptyset$\;
$\texttt{next\_vertices}\gets\emptyset;\quad \texttt{seen\_vertices} \gets \emptyset$\;

\For{$i = 1$ \KwTo $t$}{
    \ForEach{$u \in \normalfont{\texttt{current\_vertices}}$}{
        perform $d$ random neighbor queries on $u$\;
        \ForEach{edge $uw$ found by a query}{
            \If{$w \notin (\normalfont{\texttt{seen\_vertices}}\cup \normalfont{\texttt{current\_vertices}})$}{
                add $w$ to $\texttt{next\_vertices}$\;
            }
            add $uw$ to $\texttt{final\_edges}$\;
        }
        add $u$ to $\texttt{seen\_vertices}$\;
    }
    $\texttt{current\_vertices} \gets \texttt{next\_vertices}; \quad \texttt{next\_vertices} \gets \emptyset$\;
}

\Return The subgraph of $G$ induced by edges in $\texttt{final\_edges}$\;
\end{algorithm}

Hence, \bfstraverse does a breadth first search, with bounded breadth and depth, and at random.
Our final algorithm $\texttt{Tester}(n,H,\e,G)$ rejects when $|H|>|G|$. Otherwise it removes its isolated vertices from $H$ and repeats the following $n$ times (where $n$ should be independent of $G$ to ensure constant query complexity): \\
1. for every connected component~$C$ of $H$, make a call to $\bfstraverse(G,|H|,\Delta(H))$; \\
2. if the union of the subgraphs returned by the calls to \bfstraverse contain a copy of $H$, then reject. If at the end of the $n$ iterations the algorithm has not rejected $G$, then accept.

Observe that a call to \bfstraverse associated to a connected component $C$ of $H$ need not find a copy of $H$ in the algorithm. Having one call per component ensures high probability.
Concerning isolated vertices of $H$, once a copy of the non-isolated parts of $H$ is found, any other vertex of $G$ can be used as a copy of an isolated vertex of $H$. And since we have $|G|>|H|$ there are such vertices. 


\section{Testing \texorpdfstring{$H$}{H}-freeness on graphs of bounded
  expansion}\label{sec:main}

In this section, we prove our main result, Theorem~\ref{lem:testing_H_freeness}:
$H$-freeness can be tested in the random neighbor model with
constant query complexity and one-sided error on a class of graphs of
bounded expansion for any fixed graph $H$.

The proof is divided into multiple lemmas which maintain a set $\mathcal H$ of copies of $H$ in~$G$. 
Instead of analyzing the success probability of the tester on $G$, we analyze it on $G[\mathcal H]$.\linebreak
With each lemma, the set $\mathcal H$ is refined to obtain more properties.
While $\mathcal H$ shrinks with each lemma, we ensure that it remains ``sufficiently large''---so the success probability of the tester on $G[\mathcal H]$ lifts to $G$.
Essentially, the size of $\mathcal H$ is always linear in the graph size~$|G|$, where multiplicative factors depending only on $\e, |H|$ and graph structure parameters are considered constant.
At the end of this process, we obtain a set $\mathcal H$ such that $G[\mathcal H]$ has bounded treedepth. 
The main result follows then from \cite{CzumajS19}, as classes of bounded treedepth are properly minor-closed.

However, our work does not end there. 
We give a self-contained, simple-to-follow proof that $H$-freeness is testable on classes of bounded treedepth in~\Cref{sec:bounded_treedepth}.

\subsection{Reduction to subgraphs induced by many edge-disjoint
  copies of \texorpdfstring{$H$}{H}}\label{sec:reduction_subgraph}

In this section, we show that any graph $G$ that is $\e$-far from being $H$-free, has an associated set $\mathcal H$ of copies of $H$ with basic properties ensuring we can analyze the behavior of the tester on the structurally better-behaved graph $G[\mathcal H]$ instead of $G$.
More precisely, the set $\mathcal H$ has the following properties:
\begin{itemize}
	\item the copies in $\mathcal H$ are edge-disjoint,
	\item the number of copies is ``sufficiently large'', that is, linear in $|G|$, and,
	\item the degrees of a vertex of $G[\mathcal H]$ in $G[\mathcal H]$ and $G$ are of the same order of magnitude.
\end{itemize} 
The findings in this section are generalizations or distillations of results from \cite{CzumajS19} into our setting.

\begin{lemma}\label{lem:e-far-edge-disjoint} 
  Let $H$ and $G$ be graphs such that $G$ is $\e$-far from
  being $H$-free.

  There exists a set $\mathcal H$ of edge-disjoint copies of $H$ in $G$ of
  cardinality at least $\e|G|/|E(H)|$.
\end{lemma}
This lemma is a reformulation of a part of \cite[Lemma 20]{CzumajS19} and is true on arbitrary graphs.
\begin{proof}
  We prove, given a set $\mathcal H$ of $k<\e|G|/|E(H)|$
  edge-disjoint copies of $H$ in $G$, that we can build a set of $k+1$
  edge-disjoint copies of $H$ in $G$.

  Removing all edges of $G[\mathcal H]$ from $G$ removes
  $k|E(H)|<\e|G|$ edges from $G$. Since $G$ is~$\e$-far
  from being $H$-free, at least one
  copy of $H$ remains in $G-E[G[\mathcal H]]$. Adding it to $\mathcal H$ provides a set of $k+1$
  edge-disjoint copies of $H$ in $G$.
\end{proof}

We define a condition on subgraphs of $G$ and prove that the
probability of finding $H$ on a subgraph of $G$ satisfying this
condition lifts to $G$, up to a constant factor. This condition
relates to \cite[Property (a'), defined in Lemma 47]{CzumajS19}.
\begin{definition}
  Let $0<c\le 1$ be a constant. A vertex $v$ in a subgraph $G'$ of $G$ is
  called \emph{$c$-degree preserved} in $G'$ if
  $\deg^{G'}(v)\ge c\cdot \deg^G(v)$. If all vertices of $G'$ are
  $c$-degree preserved, and $|G'|\ge c|G|$, then $G'$ is called
  \emph{$c$-degree preserving}.
\end{definition}

\begin{restatable}[$\star$]{lemma}{lemHighProba}\label{lem:high-proba}
  Let $H$ and $G$ be graphs, $0<c\le 1$ a positive constant, $G'$ a
  $c$-degree preserving subgraph of $G$.
  If \normalfont{\bfstraverse}$(G',d,t)$ returns a subgraph $F$ of
  $G'$ with probability at least~$q$, then \normalfont{\bfstraverse}$(G,d,t)$ also returns $F$ with
  probability at least $c^{d^t+1}q$.
\end{restatable}

Again, this lemma is similar to \cite[Lemma 47]{CzumajS19}, the main
difference being that in \cite{CzumajS19}, the authors assume $G'$ to
be a spanning subgraph instead of having $|G'|\ge c|G|$ in the definition
of $c$-degree preserving subgraph.
See \cref{sec:appendix-missing-proof} for the detailed proof.

Finally we prove that degenerate graphs containing a set
$\mathcal H$ of many edge-disjoint copies of $H$ also contain a ``significantly large'' set
$\mathcal H'\subseteq\mathcal H$ such that $G[\mathcal H']$ is $c$-degree preserving.

\begin{restatable}[$\star$]{lemma}{lemCopyHighDegree}\label{lem:copy-high-degree}
  Let $H$ be a graph with at least one edge and $G$ a $d$-degenerate graph 
  for a constant $d>0$.
  Furthermore, let $\mathcal H$ be a set of edge-disjoint copies of a graph $H$ in $G$ such that~$|\mathcal H|\ge\alpha|G|$.
  Then, there exists $\mathcal H'\subseteq \mathcal H$ such that $G[\mathcal H']$ is $(\alpha/4d)$-degree preserving and~$|\mathcal H'|\ge\alpha|G|/2$.
\end{restatable}
This lemma generalizes \cite[Lemma 18]{CzumajS19} to graphs with bounded degeneracy.
The assumption that $E(H)>0$ is not a limitation with regard to our tester,
since $\texttt{Tester}$ treats isolated vertices separately from the rest of $H$.
The full proof is in \cref{sec:appendix-missing-proof}.

\begin{corollary} \label{cor:sec31}
  Let $H$ be a graph containing $n_H$ connected components, and let
  $G$ be a $d$-degenerate graph (for
  some constant $d>0$) that is $\e$-far from being $H$-free.

	There exists a set
	$\mathcal H$ of at least $\e|G|/2|E(H)|$ edge-disjoint
	copies of $H$ in $G$ such that if, for a positive integer $n>0$,
        $\texttt{Tester}(n,H,\e,G[\mathcal H])$ finds a copy of $H$ in
	$G[\mathcal H]$ with probability at least $2/3$, then there exists a
        positive integer $n'=O_{\e,H,n}(1)$ such that $\texttt{Tester}(n',H,\e,G)$
	does so in $G$ with probability at least
	$2/3$ in $G$.
\end{corollary}
\begin{proof}
  We define $\Hh$ with Lemmas~\ref{lem:e-far-edge-disjoint}
  and~\ref{lem:copy-high-degree}. By Lemma~\ref{lem:high-proba},
  $\texttt{Tester}(n,H,\e,G)$ finds the exact same subgraph of
  $G[\mathcal H]$ as $\texttt{Tester}(n,H,\e,G[\mathcal H])$ with
  probability $(\e/4d|E(H)|)^{nn_H(d^t+1)}$ (calls to \bfstraverse
  being independent). 
  Thus, $\texttt{Tester}(nm,H,\e,G)$ finds a
  copy of~$H$ with probability at least $p=(2/3)\cdot((\e/4d|E(H)|)^{nn_H(d^t+1)})$.
  
  Note that running $m$ times $\texttt{Tester}(n,H,\e,G)$ is the same as running $\texttt{Tester}(nm,H,\e,G)$. The probability that $\texttt{Tester}(n,H,\e,G)$ fails $m$ times in a row is $(1-p)^m$.

  Let $m=2/p$ and $n'=nm$, we conclude that $\texttt{Tester}(n',H,\e,G)$ succeeds with probability at least $1-(1-p)^m = 1-(1-2/m)^m>1-e^{-2}>2/3$.
\end{proof}

\subsection{Reduction to the case of graphs with bounded treedepth}\label{sec:from-BE-to-TD}
We first prove that in a graph $G$ which is $\e$-far from being
$H$-free and of bounded expansion, we can find a set $\mathcal H$ of many
edge-disjoint copies of $H$ such that $G[\mathcal H]$ has bounded treedepth. 
The proof relies on the existence of $p$-treedepth colorings.
We then reduce the problem of testing $H$-freeness on graphs of bounded
expansion to graphs of bounded treedepth, by combining this lemma
with the results from the previous section.
\begin{lemma}\label{lem:sub-graph-bd-tree-depth}
  Let $H$ and $G$ be graphs, with $G$ from a graph class $\mathcal C$
  of bounded expansion. Let $\mathcal H$ be a set of edge-disjoint
  copies of $H$ in $G$.

  There exists $\mathcal H' \subseteq \mathcal H$ such that
  $|\mathcal H'| = \Theta_{|H|,\mathcal C}(|\mathcal H|)$ and
  $G[\mathcal H']$ has treedepth at most $|H|$.
\end{lemma}
\begin{proof}
  Consider a $|H|$-treedepth coloring of $G[\mathcal H]$, using
  colors $c_1,\dots, c_\ell$. By
  \Cref{def:bounded_expansion},~$\ell$ is bounded by a function
  only depending on $\mathcal C$ and $|H|$. Let
  $\{C_1,\ldots,C_{\ell^{|H|}}\}$ be the set of all $|H|$-tuples of colors.
  Each copy of $H$ in $\mathcal H$ contains $|H|$ vertices, and
  therefore is colored with the colors of at least one of the $C_i$.
  Therefore, by the pigeonhole principle, there is a tuple of colors $C_i$ such that at least
  $|\mathcal H|/\ell^{|H|}$ copies of $H$ in $\mathcal H$ are colored
  with the colors of $C_i$. Let $\mathcal H'$ be the set of copies
  that only use the colors of $C_i$. 
  
  Notice that the graph $G[\mathcal H']$ is a
  subgraph of $G[C_i]$, that is, the subgraph induced by the vertices colored with colors from $C_i$ in the $|H|$-treedepth coloring. By definition of the $|H|$-treedepth coloring, $G[C_i]$ (and therefore $G[\mathcal H']$) has treedepth $|H|$. 
  We can then conclude thank to
  $|\mathcal H|/\ell^{|H|}\le |\mathcal H'|\le|\mathcal H|$,
  so $|\mathcal H'|=\Theta_{H,\mathcal C}(|\mathcal H|)$.
\end{proof}

We can now prove the reduction from bounded expansion to bounded treedepth for the success probability of the tester. 
\begin{lemma} \label{lem:reduction_tree-depth} Let $H$ and $G$ be
  graphs, with $G$ $\e$-far from being $H$-free and from a graph class
  $\mathcal C$ of bounded expansion.

  There exists a set $\mathcal H$ of edge-disjoint copies of $H$ in $G$ such
  that $G[\mathcal H]$ has bounded treedepth and, given a positive integer $n>0$, if
  $\texttt{Tester}(n,H,\e,G[\mathcal H])$ finds a copy of $H$ in~$G[\mathcal H]$
  with probability at least $2/3$, then there exists a positive integer
  $n'=O_{\e,H,n}(1)$ such that $\texttt{Tester}(n',H,\e,G)$ finds a copy of $H$ in $G$
  with probability at least $2/3$.
\end{lemma}
\begin{proof}
  We start with \cref{lem:e-far-edge-disjoint} which gives a set $\Hh_1$ of $\Omega_{\e,H}(|G|)$ many edge disjoint copies of $H$. We then apply 
  \cref{lem:sub-graph-bd-tree-depth} yielding a set $\Hh_2$ such that $G[\Hh_2]$ has treedepth at most~$|H|$, while preserving $\Hh_2 = \Omega(|G|)$.
  Finally, we apply \cref{lem:copy-high-degree} to extract a set $\Hh_3$ from $\Hh_2$ that is $c$-degree preserving for some constant $c=O_{H,\e,\Cc}(1)$, and still $\Hh_3=\Omega(|G|)$. 
  
  We conclude similarly as for \Cref{cor:sec31}: thanks to \cref{lem:high-proba}.
\end{proof}

\subsection{Testing \texorpdfstring{$H$}{H}-freeness on graphs of bounded
  expansion}\label{sec:ccl-reduction}
We are ready to prove the main result of this work.
\begin{theorem} \label{lem:testing_H_freeness} For any graph class
  $\mathcal C$ of bounded expansion, any graph $H$, and any proximity
  parameter, the property of being $H$-free can be tested in the
  random neighbor model on graphs of $\mathcal C$ with constant query
  complexity and one-sided error.
\end{theorem}
\begin{proof}
  By Lemma~\ref{lem:reduction_tree-depth}, it is sufficient to prove
  the statement when $\mathcal C$ is a class of graphs with bounded
  treedepth. This is a particular case of the main result of
  \cite{CzumajS19}, since bounded treedepth classes are properly minor-closed.
\end{proof}

In this proof, we have relied on \cite{CzumajS19} for graphs of
bounded treedepth. However, we reprove this particular case in the
next section. Hence we present a self-contained proof, that we believe
to be simpler since it avoids the machinery of edge-contractions,
hypergraphs and safety from \cite{CzumajS19}. Our proof also provides
an insight into the structure imposed by $H$-freeness on graphs of
bounded treedepth.

We believe that classes of bounded treedepth offer rich structure for
property testing and thus are a suitable ``simple'' class in the lower
ends of the sparsity hierarchy (see \Cref{pic:graphclasses}).
Most of the existing literature on property testing in the sparse
model focuses on graphs of bounded degree or on planar graphs and does
not consider classes of bounded treedepth.\footnote{A notable exception is a work by Esperet and Norin
  \cite{EsperetN22} considering treedepth to prove 1.~a linear
  Erdos--Posa property for monotone properties on properly
  minor-closed classes and 2.~that proper minor-closed classes admit
  approximate proof labelling schemes of logarithmic complexity.}


\section{Testing \texorpdfstring{$H$}{H}-freeness for graphs with bounded treedepth}
\label{sec:bounded_treedepth}

We start by explaining how we proceed before giving the formal definitions of the terms used in the following introduction.

Consider the following family of graphs (a graph for each $n>0$, similar to \cite[Figure
3]{CzumajMOS19}):
\begin{center}
  \includegraphics[width=0.5\linewidth]{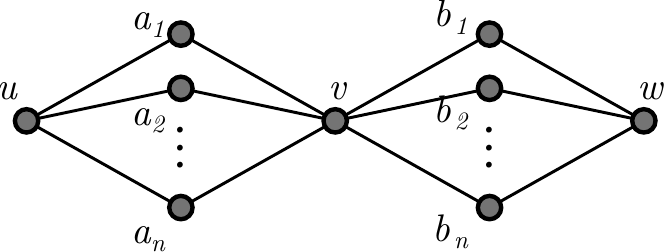}
\end{center}
The maximum length of its paths is five, and its
treedepth at most five as well. Furthermore, it contains $n$ edge-disjoint
copies of $P_5$ of the shape $ua_ivb_iw$. Write $\Hh$ for the
corresponding set of copies of $P_5$. Consider a run of \bfstraverse.
When $n\to \infty$, the probability of finding a copy of $P_5$ which
is in $\mathcal H$ goes to zero. However, with breadth one and depth
four, \bfstraverse finds a copy of $P_5$ of shape $ua_ivb_jw$ with high
probability, independently of the value of $n$.

\bfstraverse also finds other kinds of copies of $P_5$ with high
probabilities: $a_ivb_jwb_k$ for example. However, we focus on the copies of $P_5$
which are related to those in $\Hh$, i.e. that agree on common
vertices (the vertices $u$, $v$, and $w$): each vertex of $G$
corresponds to exactly one vertex of $H$ via copies in $\Hh$. We say that $G$ is
\emph{uniformly colored} by the vertices of $H$.

The particularity of copies of $P_5$ of shape $ua_ivb_jw$ is that they
consist in reconstructed copies of $H$ from \emph{parallel parts} of
$H$ (reduced here to the vertices $a_i$ and $b_j$) attached to a few number of
vertices (the vertices $u$, $v$, and $w$) we call \emph{sources}.
The specific choices of \emph{parallel parts} does not matter as long as
some \emph{compatibility} is ensured on the sources.

In this section, we show that this behavior is not specific to copies
of $P_5$ in the family of graphs introduced above, and we extend
these ideas to graphs of bounded treedepth.

At first, we introduce formally the notions of \emph{uniformly colored}
set of copies of $H$. Furthermore, we relate the copies of $H$ to the
treedepth via \emph{uniformly layered} treedepth embedding. We
show that testing on graphs of bounded treedepth can be reduced to
testing on graphs of treedepth at most $|H|$ with uniformly layered
embeddings.

Then, we formalize the notions of parallel parts, sources,
compatibility of parts, and prove that in a graph of treedepth at most
$|H|$ with a uniformly layered embedding, copies of~$H$ can be
constructed somewhat blindly using parallel parts, by choosing
via which copy of a source we enter the next parallel part.

At last, we show that such constructions of connected components of $H$ can be simulated in
the \bfstraverse assuming a breadth of $\Delta(H)$ and a depth of
$|H|$, and with probability~$\Omega_H(1)$. As a consequence, we
conclude that $H$-freeness can be tested in the random neighbor
model with constant query complexity and one-sided error on
classes of graphs of bounded treedepth in the case $H$ is connected.
Finally, we extend this result to nonnecessarily connected graphs~$H$
by using uniformly layered embedding. Indeed, such embeddings
ensure that with high probability, we find disjoint copies of the
connected components of $H$ during successive runs of \bfstraverse.

\subsection{Reduction to uniformly layered treedepth embeddings}
\begin{definition}\label{def-colored}
  Let $H$ and $G$ be two graphs.

  A set $\Hh$ of copies of $H$ in $G$ is called \emph{uniformly
    colored} (by the vertices of $H$) if for all copies $h, h' \in \Hh$ and vertices
  $a,a' \in V(H)$, having $h(a)=h'(a')$ implies $a=a'$.
\end{definition}

Intuitively, if two copies overlap on some vertex $v$ of $G$, they agree on the vertex $a$ of $H$ that is paired with $v$.
Observe that the next lemma is a reformulation of a part of \cite[Lemma~20]{CzumajS19} and is true on arbitrary graphs.
\begin{lemma}
  \label{lem:many-uniformly-colored-copies}
  Let $\Hh$ be a set of $N$ edge-disjoint copies of a graph $H$
  in a graph $G$.

  There exists a uniformly colored set
  $\Hh'\subseteq\Hh$ of cardinality at least
  $N/|H|^{|H|}$.
\end{lemma}
\begin{proof}
  We want the inverse image of a given vertex of $G$ under all copies
  of $\Hh'$ to be either empty or a singleton. We consider
  random assignements of functions $f\colon V(G)\to V(H)$ and look at
  the probability that the copies of $H$ in $\Hh$ are
  \emph{compatible with $f$}, meaning that for every $v\in V(H)$,
  $f(h(v))=v$.

  Write $\{h_1,\dots,h_N\}$ for $\Hh$ and $X_i$ for the random
  variable which is $1$ if $f$ is compatible with $h_i$ and $0$
  otherwise. Since:
  \begin{itemize}
  \item the probability of $f(h_i(v))=v$ for a given $v$ is
    $1/|H|$,
  \item the probabilities of $f(h_i(v))=v$ are pairwise independent
    for distinct values of $v$, and,
  \item having $X_i=1$ is exactly having $f(h_i(v))=v$ for all
    $v\in V(H)$,
  \end{itemize}
  the expected value of $X_i$ is $1/|H|^{|H|}$.
  The number of copies compatible with $f$ is given by $\sum_iX_i$,
  and its expected value is
  \[
    \mathbb E[X]=\sum_i\mathbb E[X_i]=\sum_i\mathbb
    P[X_i=1]=\frac{N}{|H|^{|H|}}.
  \]
  Therefore, there must exist some $f$ with at least this many uniformly colored copies of~$H$.
\end{proof}

\begin{definition}
  Let $\Hh$ be a set of copies of a graph $H$ in a graph $G$.

  A tree embedding $\le_T$ of $G[\Hh]$ is called
  \emph{uniformly-layered} with regard to $\Hh$ if $\Hh$
  is uniformly colored and for every color (with regard to
  $\Hh$), there is a level of $\le_T$ containing precisely
  the vertices of this color.
\end{definition}

\begin{lemma}\label{lem:many-layered-copies}
  Let $H$ and $G$ be graphs, with $G$ of treedepth $d$.
  Let $\Hh$ be a uniformly-colored set of edge-disjoint copies
  of $H$ in $G$.
  There exists $\Hh' \subseteq \Hh$ such that
  $|\Hh'| = \Theta_{|H|,d}(|\Hh|)$ with a
  uniformly-layered tree embedding~$\le_T$ of
  $G[\Hh']$ with depth $|H|$.
\end{lemma}
\begin{proof}
  Write $k$ for $|H|$ and
  $a_1,\dots,a_k$ for the vertices of $H$. Let $G$ be a graph of
  treedepth~$d$, and $\Hh$ be a set of edge-disjoint and
  uniformly colored (by the vertices $a_1,\dots a_k$) copies of~$H$ in
  $G$. Let~$\le_T$ be a tree embedding of $G$ with depth~$d$. 
  Let $C_1,\ldots,C_{d^k}$ be the set of all $k$-tuples of levels. For each
  copy $h$ of $H$ in $\Hh$, and for each $i\le k$, there is a
  level $\ell_i$ such that~$h(a_i)$ lies in level $\ell_i$ in $\le_T$: to $h$
  we associate the tuple $C=(\ell_1,\ldots,\ell_k)$.

  Since there are $d^k$ distinct tuples, there exists a $k$-tuple $C$
  associated to at least $|\Hh|/d^k$ copies. Let $\Hh'$ be the
  subset of $\Hh$ of copies associated to $C$. By construction,
  we have:\linebreak 
  1. $|\Hh'|\ge |\Hh|/d^k$, and \\
  2. $\Hh'$ is uniformly layered\footnote{An astute reader can notice that if $\ell_i=\ell_j$ for some pair $i,j$ then one level can contain two or more colors. This is easily fixable by creating intermediate levels (e.g.~3.1, 3.2, $\ldots$) so that vertices of several colors in one level can choose arbitrarily their sub-level. E.g.~if vertices blue and red are on level 3, we arbitrary put the blue vertices on level 3.1, and the red one in 3.2. This keeps at most $H$ many (sub-)levels.} 
  in $\le_T^{\Hh'}$ (that is $\le_T$ restricted to the vertices of $G[\Hh']$).
\end{proof}

By combining this result with the lemmas from
Section~\ref{sec:reduction_subgraph},
testing $H$-freeness on graphs
of treedepth at most $d$ is reduced to testing $H$-freeness on graphs
induced by many edge-disjoint copies of $H$, of treedepth at most
$|H|$, with a uniformly layered treedepth embedding of depth
$|H|$. We are more precise on the combined use of these lemmas in \cref{sec:combining-bd-td}.

\subsection{Constructing copies of \texorpdfstring{$H$}{H} via parallel parts}\label{sec:algo-on-TD}
\begin{definition}\label{def:source}
  Let $H$ be a graph with a total ordering $\pi$ on its vertices.

  A vertex $v$ of $H$ is called a \emph{source} if it is not adjacent
  to any of its predecessors in $\pi$. Otherwise, it is referred to as
  an \emph{inner vertex}.
\end{definition}

Observe, in particular, that the minimal vertex in $\pi$ is a source.

\begin{definition}
  Let $H$ be a graph with a total ordering $\pi$ on its vertices.
  
  A \emph{parallel part} of $H$ is a maximal connected subgraph of $H$ that contains
  no source.
  For convenience, the terms parallel part and part are used interchangeably.
\end{definition}

\begin{observation}
  Any graph $H$ with a total vertex ordering $\pi$ has a unique vertex
  partition into sources and parallel parts. Moreover, the set of sources of $H$ is an independent set in $H$.
\end{observation}

Let $H$ and $G$ be graphs, with $G$ induced by the edges of a set
$\Hh$ of edge-disjoint copies of $H$, and with a uniformly layered
treedepth embedding of $G$ of depth $|H|$.

We now consider the vertex ordering on $H$ induced by the levels of the
uniformly layered treedepth embedding of $G$.
Observe that if $h\in\Hh$ and $v$ is in a parallel part of $H$, then for every color, $h(v)$ has at most one neighbor of this specific color.
Therefore, as the copies in~$\Hh$ are edge-disjoint, we get the following observation.
\begin{observation}
  Let $P$ be a parallel part of $H$ and $h$ a copy of $H$ in
  $\Hh$.

  Every vertex in $h(P)$ is contained in exactly one copy of $H$ in
  $\Hh$. In particular, if a vertex $v$ of $G[\Hh]$ is
  contained in multiple copies from $\Hh$, then $v$ is a (copy
  of a) source.
\end{observation}

\begin{definition}
  Let $F_1$ and $F_2$ be two parallel parts of $H$ and $h_1$ and $h_2$
  be two copies of $H$ in $G$. We say that $h_1(F_1)$ and $h_2(F_2)$ are
  \emph{compatible} if for every source $s$ adjacent to both
  $F_1$ and $F_2$, we have $h_1(s)=h_2(s)$.
\end{definition}

We now show how copies of $H$ may be constructed from compatible
parallel parts and their adjacent sources:
\begin{observation}\label{obs:compatible_parts}
  Let $F_1,\dots,F_n$ be the parts of $H$ and $S_1,\dots,S_n$ the sets
  of sources they are adjacent to, respectively. Given
  $h_1,\dots,h_n\in\Hh$ copies of $H$ such that the graphs~$h_i(F_i)$
  are pairwise compatible, the graph $\cup_ih_i(F_i\cup S_i)$ is a
  copy of $H$ in $G$ (which may not be in~$\Hh$).
\end{observation}

We now introduce the notion of port. Informaly, a port is a copy of a
source from which a copy of a yet to be discovered parallel part of
$H$ may be found.
\begin{definition}
  Let $\Ff=\{F_i\mid i\in I\}$ be a set of distinct parts of $H$.
  A source $s$ adjacent to $\cup_iF_i$ is called a \emph{port} of
  $\Ff$ if there exists a part $F\notin\Ff$ of $H$ which is adjacent
  to $s$.
\end{definition}

We show that, assuming $H$ is connected and starting with the copy of
a parallel part of~$H$, we can always reconstruct a copy of $H$
``blindly'', meaning that each time we consider the copy of a new part
of $H$, we only choose the port from which we discover it.

Consider the following procedure:

\begin{algorithm}[H]
\caption{Constructing a copy of a connected component of $H$}
\label{alg:constructing_H}
\KwIn{A graph $G[\Hh]$ induced by a set $\Hh$ of edge-disjoint copies.

A uniformly layered tree embedding $\le_T$ of depth $|H|$ of $G$.

A copy $h(F)$ of a parallel part of $H$ contained in the component $C_F$ of $H$.}
\KwOut{A set $\Ff$ of copies of parts of $H$ (including $h(F)$) forming a copy of $C_F$}

$\Ff \gets \{h(F)\}$\;

\While{there exists a parallel part $F'$ of $C_F$ without a copy in $\Ff$}{
    $h'(F')\gets$ a copy of a parallel part of $C_F$ without a copy in $\Ff$, and adjacent to a port of $\Ff$ of maximal level in $\le_T$\;
    $\Ff \gets \Ff ~\cup~ \{h'(F')\}$\;
}

\Return $\Ff$\;
\end{algorithm}

\pagebreak
\begin{lemma}
  \label{lem:main_algo}
  Algorithm~\ref{alg:constructing_H} is sound and terminates.
\end{lemma}
\begin{proof}
  Termination is ensured by the fact that $H$, and by extension $C_F$,
  contain finitely many parallel parts. At the end of a run, $\Ff$
  contains precisely one copy of each part of $C_F$, including $h(F)$.

  We prove soundness via the following invariant: at any step during a
  run of Algorithm~\ref{alg:constructing_H}, the set $\Ff$
  contains pairwise compatible copies of parts of $C_F$.
  We then conclude with \Cref{obs:compatible_parts}.

  We proceed inductively over the while loop. Let $\Ff$ be
  given after a partial execution of the algorithm. Since the
  algorithm is still running, there exists a part of $C_F$ without a
  copy in $\Ff$. Let $h'(F')\in\Hh$ be a copy of a
  parallel part $F'$ of $C_F$ without a copy in $\Ff$, and such that $F'$ is adjacent to
  a port $s$ of $\Ff$ of maximal level in $\le_T$.

  Assume for the sake of
  contradiction that $h'(F')$ is not compatible with
  some copy $h''(F'')\in \Ff$: let $s'$ be a source
  of $C_F$ adjacent with both $F'$ and $F''$ but with
  $h'(s')\neq h''(s')$.

  By construction of the algorithm, the graph $G'$ induced by the copies of parts
  in $\Ff\cup\{h'(F')\}$ as well as by the corresponding copies of
  adjacent sources is connected.

  Consider, in $G'$, a simple path $P$
  from $h'(s')$ to $h''(s')$ going through copies of parts of~$H$ via the
  ports that led to their addition to $\Ff$. Since initially $\Ff=\{h(F)\}$, such a path is well-defined.

  Let $s_m$ be the copy of a
  source of smallest level in $\le_T$ appearing on $P$. At $s_m$,
  $P$ must change of parallel part. It quits a copy of a part $F_1$ to
  enter a copy of a part $F_2$.

  We perform case analysis depending on whether $F_1$ is discovered
  before $F_2$ by the algorithm or not.

  Assume that $F_2$ is discovered after $F_1$ by the algorithm.
  Consider the sequence of ports that led to discoveries of parts traversed by
  $P$: $h'(s')=s_1,\dots,s_i,s_m,s_1',\dots,s_j'=h''(s'')$.

  By construction of $P$, all sources among $s_1,\dots,s_i$ are at a deeper level in
  $\le_T$ than $s_m$. Hence, all copies of parts of $H$ that the path 
  $h'(s')=s_1,\dots,s_i,s_m$
  traverses, as well as~$h'(F')$, should have been discovered before
  the copy of
  $F_2$, a contradiction since $h'(F')$ is discovered
  after the copy of $F_2$.

  Assume that $F_1$ is discovered after $F_2$ by the algorithm.
  Similarly, as in the previous case, all copies of parts of $H$ that the path
  $s_m,s_1',\dots,s_j'=h''(s'')$
  traverses, as well as a copy of~$F'$ adjacent to $h''(s')$
  should have been discovered before a copy of $F_1$,
  a contradiction since no copy of $F'$ adjacent to $h''(s')$ has been discovered.
\end{proof}

\subsection{Testing \texorpdfstring{$H$}{H}-freeness on graphs of bounded treedepth}\label{sec:combining-bd-td}
To prove testability of $H$-freeness on graphs of bounded treedepth we
proceed as follows. First, we prove that
Algorithm~\ref{alg:constructing_H} can be simulated inside a run of
\bfstraverse. Hence, the latter (with the right parameters) finds a
copy of a given connected component of $H$ with high probability.
Then, by composing such calls sequentially as is done in
\texttt{Tester}, we find $H$ as a whole with high probability. Such a
sequential composition is possible thanks to uniformly colored sets of
copies. Indeed, the copies of connected components of $H$ found by
\bfstraverse may use parallel parts of different copies, but are
always compatible with the uniform coloring of copies of $H$, ensuring
that found copies are disjoint. Finally, to ensure that we find $H$
with probability $2/3$ and not $\Omega_{|H|,\e}(1)$, we use the
parameter $n$ of \texttt{Tester} to strengthen probabilities.

Note that, the intermediate lemmas do not require $G$ to have bounded treedepth, but~$G[\Hh]$ to have a uniformly layered treedepth embedding. Thanks to \cref{lem:many-layered-copies} the former implies the latter.

\begin{lemma}\label{lem:proba_neighborhood}
  Let $H$ and $G$ be graphs, $\Hh$ a set of edge-disjoint copies of a graph $H$, and~$\le_T$
  be a uniformly layered treedepth embedding of $G[\Hh]$ of depth $|H|$.

  Consider a run of \bfstraverse (on $G[\Hh]$) with breadth $\Delta(H)$ and depth
  $t$. If, at step $n<t$ the algorithm has found the copy $h(u)$ of an
  inner vertex $u\in V(H)$ ($h\in\Hh$), then with probability
  $\Omega_{H}(1)$, the copy $h(N_H(u))$ of the neighborhood of $u$ is
  discovered at step~$n+1$.
\end{lemma}
\begin{proof}
  Since $u$ is an inner vertex and $\le_T$ a uniformly layered
  treedepth embedding, the neighborhood of $h(u)$ in $G[\Hh]$ is
  precisely the image of $N_H(u)$ under $h$. Hence, in the step just
  after discovering $u$, \bfstraverse finds all edges incident to $u$
  with probability at least $!(\Delta_H-1)/\Delta_H^{\Delta_H-1}=\Omega_H(1)$. 
\end{proof}

\begin{lemma}\label{lem:proba_sources}
  Let $H$ and $G$ be graphs, $\Hh$ a set of edge-disjoint copies of a graph $H$, and~$\le_T$
  be a uniformly layered treedepth embedding of $G[\Hh]$ of depth $|H|$.

  Let $s$ be a source of $H$, $F$ a parallel part of $H$ adjacent to
  $s$.

  Consider a run of \bfstraverse(on $G[\Hh]$) with breadth $\Delta(H)$ and depth $t$.
  If, at step $n<t$ the algorithm has found a copy $h(s)$ of $s$
  ($h\in\Hh$), then with probability $\Omega_H(1)$, there
  exists a copy $h'\in\Hh$ with $h'(s)=h(s)$ such that a copy
  of an inner vertex of $F$ under~$h'$ is discovered at
  step $n+1$.
\end{lemma}
\begin{proof}
  Let $F$ and $F'$ be two distinct parallel parts, both adjacent to a
  source $s$ of $H$, and let~$h$ be a copy of $H$ in $\Hh$.
  Observe that there are as many copies of $F$ as there are of $F'$
  which are adjacent to $h(s)$ in $G$ (one each per copy with $h(s)$ in the image).

  Let $N$ be the number of copies of $H$ $h(s)$ is in the image of. Let
  $m_F$ and $m_{F'}$ be the number of edges of $H$ in $F$ and $F'$,
  respectively, which are incident to $s$. The probability
  that one step after discovering $h(s)$, \bfstraverse
  finds edges from $h(s)$ to vertices copies of inner vertices of $F$
  is at least $1-(m_F/\deg^H(u))^{\Delta_H}=\Omega_H(1)$.
\end{proof}

\begin{lemma}\label{lem:proba_parts1}
  Let $H$ and $G$ be graphs, $\Hh$ a set of edge-disjoint copies of a graph $H$, and~$\le_T$
  be a uniformly layered treedepth embedding of $G[\Hh]$ of depth $|H|$.

  Let $F$ be a parallel part of $H$,
  and $u$ an inner vertex of $F$.

  Consider a run of \bfstraverse (on $G[\Hh]$) with breadth $\Delta(H)$ and depth
  $t>|F|$. If, at step~$n\le t-(|F|-1)$ the algorithm has already
  discovered the copy $h(u)$ of $u$ ($h\in\Hh$), then with probability
  $\Omega_{H}(1)$, the copy $h(F)$ as well as the copies
  of the sources adjacent to~$F$ under $h$ are discovered at step
  $n+|F|-1$. 
\end{lemma}
\begin{proof}
  By applying Lemma~\ref{lem:proba_neighborhood} at most $|F|$ times,
  since $|F|$ bounds the diameter of $F$ from above. 
\end{proof}

\begin{lemma}\label{lem:proba_parts2}
  Let $H$ and $G$ be graphs, $\Hh$ a set of edge-disjoint copies of a graph $H$, and~$\le_T$
  be a uniformly layered treedepth embedding of $G[\Hh]$ of depth $|H|$.

  Let $u$ be an inner vertex of a connected component $C$ of $H$.

  Consider a run of \bfstraverse (on $G[\Hh]$) with breadth $\Delta(H)$ and depth
  $t>|C|$. If the algorithm starts with a copy $h(u)$ of $u$
  ($h\in\Hh$), then with probability $\Omega_{\H}(1)$, the
  algorithm finds a copy of $C$ consisting of copies of its parallel
  parts corresponding to elements of~$\Hh$. 
\end{lemma}
\begin{proof}
  \Cref{lem:proba_sources,lem:proba_parts1} makes it
  possible to simulate the behavior of
  Algorithm~\ref{alg:constructing_H} inside a run of the \bfstraverse. With
  breadth $\Delta(H)$ and depth $|C|$ and probability~$\Omega_H(1)$,
  the output of \bfstraverse contains a copy of $C$ in $G$ which is constructed just as
  in the algorithm. Transitions between parts is ensured by \Cref{lem:proba_sources}.
\end{proof}

\begin{theorem} \label{lem:testing_H_freeness_treedepth} For any graph
  $H$ and graph class $\mathcal C$ of bounded treedepth, the
  property of being $H$-free can be tested in the random neighbor
  model on graphs of $\mathcal C$ with constant query
  complexity and one-sided error.
\end{theorem}
\begin{proof}
  Fix a graph $G\in\mathcal C$ which is $\e$-far from being $H$-free.
  We prove
  that there exists a constant $n=O_{\e,H,\Cc}(1)$ such that 
  $\texttt{Tester}(n,H,\e,G)$
  rejects with probability at least $2/3$. Since $\texttt{Tester}$ starts
  by stripping isolated vertices from~$H$, we assume that $H$ contains none.
   Let $n_H$ be the number of connected components of $H$. 
  First, we use \Cref{lem:e-far-edge-disjoint,lem:many-uniformly-colored-copies,lem:many-layered-copies,lem:copy-high-degree} (in this order) to 
  obtain a set $\Hh$ of cardinality $\Omega_{\e,H,\Cc}(|V(G)|)$ such that $G[\Hh]$ is $c$-degree preserving and contains a uniformly-layered tree embedding\footnote{Note that having uniformly-layered tree embedding is a monotone property, so restricting further to get $c$-degree preserving is safe.} (where $c=O_{\e,H,\Cc}(1)$).
  Let $C$ be a connected component of~$H$. We prove that the output of a single call
  to \bfstraverse$(G[\Hh],|H|,\Delta(H))$ 
  contains a copy of $C$ compatible with
  the uniformly-layered tree embedding with probability
  $\Omega_{\e,H}(1)$. By \Cref{lem:proba_parts2}, it is sufficient to
  prove that \bfstraverse starts on the copy of an inner vertex of
  $C$ with high probability. Since $C$ contains an edge, it cannot be reduced to a single source, and
  $C$ contains at least one inner vertex. As copies of inner vertices
  are in the image of precisely one copy of $H$, and as~$G[\Hh]$ contains
  $\Omega_{\e,H,\Cc}(|G|)$ copies of $H$, $G[\Hh]$ contains $\Omega_{\e,H,\Cc}(|G|)$
  vertices copies of inner vertices of~$C$. Hence, 
  \bfstraverse$(G[\Hh],|H|,\Delta(H))$ 
  starts on an inner vertex of $C$ with probability~$\Omega_{\e,H,\Cc}(1)$.

  Now, with \cref{lem:high-proba}, and the fact that $\Hh$ is $c$-degree preserving, we have that the output of a single call to \bfstraverse$(G,|H|,\Delta(H))$ 
  contains a copy of $C$ compatible with
  the uniformly-layered tree embedding with probability
  $\Omega_{\e,H,\Cc}(1)$ as well.

  Each execution of the part of the algorithm which is repeated $n$
  times (i.e. the $n_H$ calls to \bfstraverse) finds a copy of $H$
  with probability $p=\Omega_{\e,H,\Cc}(1)$ each time it is executed.
  Indeed, the $i$th connected component is found with probability
  $\Omega_{\e,H,\Cc}(1)$ by the $i$th call to \bfstraverse, in the form of
  a copy compatible with the uniformly-layered tree embedding, so
  these copies of the connected components of $H$ found with high
  probability are vertex disjoints and form a copy of $H$. 
  Let $n$ be such that $2/n=p$.
  Then, the probability that $\texttt{Tester}(n,H,\e,G)$ finds
  a copy of $H$ is at least $1-(1-2/n)^{n}>1-e^{-2}>2/3$.
\end{proof}


\section{Outlook}

Could one extend our results to broader classes of graphs? At least not in the direction of bounded degeneracy. In \cite{AwofesoGLR25} is it proved that for any $r\ge 4$, $C_r$-freeness is not testable for graphs of bounded $(\lfloor r/2 \rfloor - 1)$-admissibility. This being a special case of bounded degeneracy.

However, bounded $r$-admissibility is incomparable with nowhere denseness.
Nowhere dense is a very robust notion that yield numerous algorithms in recent years. A key difference with bounded expansion, is that while nowhere dense graphs enable $p$-treedepth coloring, it requires $O_{\delta,p}(|G|^\delta)$ many colors (and not constantly many). Meaning that constant numbers in some of our key lemmas would now depend on $G$. Additionally, nowhere dense graphs can have a super-linear number of edges (e.g.~$|G|\cdot\log(|G|)$). One could start by adapting the concept of $\epsilon$-far, enabling removal of an edge set with size $O_\delta(\epsilon\cdot |G|^{1+\delta})$ for all $\delta >0$. 

Another extension, as in \cite{EsperetN22}, could be to look at approximate
proof labelling scheme. However, the main result of~\cite{EsperetN22} (pushing from finitely many forbidden subgraphs to monotone properties) cannot be adapted to bounded expansion. In part because already bipartiteness is not testable on bounded degree graphs, so even less so on graphs with bounded expansion.

\newpage
\bibliography{biblio.bib}

@string{FOCS    = "Symp. on Foundations of Computer Science (FOCS)"}

@string{ICALP	  = "Intl. Coll. on Automata, Languages and Programming (ICALP)"}

@string{STACS	= "Symp. on Theoretical Aspects in Computer Science (STACS)"}

@book{Goldreich17,
  author       = {Oded Goldreich},
  title        = {Introduction to Property Testing},
  publisher    = {Cambridge University Press},
  year         = {2017},
  url          = {http://www.cambridge.org/us/catalogue/catalogue.asp?isbn=9781107194052},
  doi          = {10.1017/9781108135252},
  isbn         = {978-1-107-19405-2},
  bibsource    = {dblp computer science bibliography, https://dblp.org}
}

@InProceedings{AwofesoGLR25,
  author       = {Christine Awofeso and
                  Patrick Greaves and
                  Oded Lachish and
                  Felix Reidl},
  title        = {Results on H-Freeness Testing in Graphs of Bounded r-Admissibility},
  booktitle    = STACS,
  series       = {LIPIcs},
  volume       = {327},
  pages        = {12:1--12:16},
  publisher    = {Schloss Dagstuhl - Leibniz-Zentrum f{\"{u}}r Informatik},
  doi          = {10.4230/LIPIcs.STACS.2025.12},
  year         = {2025}
}

@inproceedings{AwofesoGLLR25,
  author       = {Christine Awofeso and
                  Patrick Greaves and
                  Oded Lachish and
                  Amit Levi and
                  Felix Reidl},
  title        = {Testing {C}k-Freeness in Bounded Admissibility Graphs},
  booktitle    = ICALP,
  series       = {LIPIcs},
  volume       = {334},
  pages        = {15:1--15:20},
  publisher    = {Schloss Dagstuhl - Leibniz-Zentrum f{\"{u}}r Informatik},
  doi          = {10.4230/LIPICS.ICALP.2025.15},
  year         = {2025}
}

@article{AdlerKP24,
  author       = {Isolde Adler and
                  Noleen K{\"{o}}hler and
                  Pan Peng},
  title        = {On Testability of First-Order Properties in Bounded-Degree Graphs
                  and Connections to Proximity-Oblivious Testing},
  journal      = {{SIAM} J. Comput.},
  volume       = {53},
  number       = {4},
  pages        = {825--883},
  year         = {2024},
  url          = {https://doi.org/10.1137/23m1556253},
  doi          = {10.1137/23M1556253},
  bibsource    = {dblp computer science bibliography, https://dblp.org}
}

@article{NewmanS13,
  author       = {Ilan Newman and
                  Christian Sohler},
  title        = {Every Property of Hyperfinite Graphs Is Testable},
  journal      = {{SIAM} J. Comput.},
  volume       = {42},
  number       = {3},
  pages        = {1095--1112},
  year         = {2013},
  url          = {https://doi.org/10.1137/120890946},
  doi          = {10.1137/120890946},
  bibsource    = {dblp computer science bibliography, https://dblp.org}
}

@article{GoldreichR02,
  author       = {Oded Goldreich and
                  Dana Ron},
  title        = {Property Testing in Bounded Degree Graphs},
  journal      = {Algorithmica},
  volume       = {32},
  number       = {2},
  pages        = {302--343},
  year         = {2002},
  url          = {https://doi.org/10.1007/s00453-001-0078-7},
  doi          = {10.1007/S00453-001-0078-7},
  bibsource    = {dblp computer science bibliography, https://dblp.org},
  note         = {definition bounded degree model},
}

@article{AlonFNS09,
  author       = {Noga Alon and
                  Eldar Fischer and
                  Ilan Newman and
                  Asaf Shapira},
  title        = {A Combinatorial Characterization of the Testable Graph Properties:
                  It's All About Regularity},
  journal      = {{SIAM} J. Comput.},
  volume       = {39},
  number       = {1},
  pages        = {143--167},
  year         = {2009}
}

@article{AlonS08,
  author       = {Noga Alon and
                  Asaf Shapira},
  title        = {A Characterization of the (Natural) Graph Properties Testable with
                  One-Sided Error},
  journal      = {{SIAM} J. Comput.},
  volume       = {37},
  number       = {6},
  pages        = {1703--1727},
  year         = {2008}
}

@article{AlonS08a,
  author       = {Noga Alon and
                  Asaf Shapira},
  title        = {Every Monotone Graph Property Is Testable},
  journal      = {{SIAM} J. Comput.},
  volume       = {38},
  number       = {2},
  pages        = {505--522},
  year         = {2008}
}

@article{GoldreichGR98,
  author       = {Oded Goldreich and
                  Shafi Goldwasser and
                  Dana Ron},
  title        = {Property Testing and its Connection to Learning and Approximation},
  journal      = {J. {ACM}},
  volume       = {45},
  number       = {4},
  pages        = {653--750},
  year         = {1998}
}

@article{NesetrilM08,
  author       = {Jaroslav Nesetril and
                  Patrice Ossona de Mendez},
  title        = {Grad and classes with bounded expansion~{I.} {D}ecompositions},
  journal      = {Eur. J. Comb.},
  volume       = {29},
  number       = {3},
  pages        = {760--776},
  year         = {2008},
  url          = {https://doi.org/10.1016/j.ejc.2006.07.013},
  doi          = {10.1016/J.EJC.2006.07.013},
  bibsource    = {dblp computer science bibliography, https://dblp.org}
}

@book{sparsity,
  author    = {Jaroslav Nešetřil and
               Patrice Ossona de Mendez},
  title     = {Sparsity - Graphs, Structures, and Algorithms},
  series    = {Algorithms and combinatorics},
  volume    = {28},
  publisher = {Springer},
  year      = {2012},
  url       = {https://doi.org/10.1007/978-3-642-27875-4},
  doi       = {10.1007/978-3-642-27875-4},
  isbn      = {978-3-642-27874-7},
  bibsource = {dblp computer science bibliography, https://dblp.org}
}

@article{CzumajMOS19,
  author       = {Artur Czumaj and
                  Morteza Monemizadeh and
                  Krzysztof Onak and
                  Christian Sohler},
  title        = {Planar graphs: Random walks and bipartiteness testing},
  journal      = {Random Struct. Algorithms},
  volume       = {55},
  number       = {1},
  pages        = {104--124},
  year         = {2019},
  url          = {https://doi.org/10.1002/rsa.20826},
  doi          = {10.1002/RSA.20826},
  bibsource    = {dblp computer science bibliography, https://dblp.org}
}

@inproceedings{CzumajS19,
  author       = {Artur Czumaj and
                  Christian Sohler},
  title        = {A Characterization of Graph Properties Testable for General Planar
                  Graphs with one-Sided Error (It's all About Forbidden Subgraphs)},
  booktitle    = FOCS,
  year         = {2019},
  doi          = {10.1109/FOCS.2019.00089}
}

@inproceedings{EsperetN22,
  author       = {Louis Esperet and
                  Sergey Norin},
  title        = {Testability and Local Certification of Monotone Properties in Minor-Closed
                  Classes},
  booktitle    = ICALP,
  series       = {LIPIcs},
  volume       = {229},
  pages        = {58:1--58:15},
  publisher    = {Schloss Dagstuhl - Leibniz-Zentrum f{\"{u}}r Informatik},
  doi          = {10.4230/LIPICS.ICALP.2022.58},
  year         = {2022}
}

\newpage
\appendix
\section{Missing proofs}\label{sec:appendix-missing-proof}

In this appendix, we formally prove \cref{lem:high-proba,lem:copy-high-degree}. Both are inspired from~\cite{CzumajS19}. 

\lemHighProba*

\begin{proof}
  Let $F$ be a subgraph of $G'$ which
  is returned by $\bfstraverse(G',d,t)$ with probability $q$.

  Assume the search starts on a vertex $v$ of $G'$.

  Since $G'$ is $c$-degree preserving, $\bfstraverse$ $(G,d,t)$ starts
  on $v$ with probability $1/|G|\ge c/|G'|$ (this is the $+1$
  in $c^{d^t+1}$).

  Let $uw$ be an edge of $G'$. Assume, during runs of
  $\bfstraverse(G',d,t)$ and $\bfstraverse(G,d,t)$ both starting on
  $v$, that $u$ has just been visited, and $uw$ has not been
  discovered yet by either runs.

  Since $G'$ is $c$-degree preserving, the neighborhood of $u$ in $G$
  is larger than in $G'$ by a factor of at most $1/c$. The
  probabilities that the runs of $\bfstraverse(G',d,t)$ and
  $\bfstraverse(G,d,t)$, just after visiting $u$, discover
  $uw$ are, respectively:
  \[
    \frac{1}{\deg^{G'}(u)}~~~~\text{and}~~~~\frac{1}{\deg^{G}(u)}\ge\frac{c}{\deg^{G'}(u)}.
  \]
  In other words, the probability that the run of $\bfstraverse(G,d,t)$
  discovers $uw$ just after visiting $u$ is at least $c$ times the
  probability that the run of $\bfstraverse(G',d,t)$ discovers $uw$ just after
  visiting $u$.

  By construction, there are at most $d^t$ edges in $F$: the
  probability that $\bfstraverse$ $(G,d,t)$ discovers precisely the
  edges of $F$, and in the same order than
  $\bfstraverse$ $(G',d,t)$ does, is at least $c^{d^t}q$.

  Combining this with the probability that $\bfstraverse(G,d,t)$
  starts on $v$, the probability that $\bfstraverse(G,d,t)$ finds a
  copy of $F$ in $G$ is at least $c^{d^t+1}q$.
\end{proof}

\lemCopyHighDegree*

\begin{proof}
  We prove this statement by providing a simple algorithm computing
  such a set~$\mathcal H'$: take $\cal H'=\cal H$ at first. While
  there exists a vertex of $G[\cal H']$ which is not
  $(\alpha/4d)$-degree preserving in $G[\cal H']$, choose one, say
  $v$, and remove from $\cal H'$ every copy of $H$ containing $v$ in
  their image. This procedure starts with $\mathcal H' = \mathcal H$
  and ends when $G[\mathcal H']$ is $(\alpha/4d)$-degree preserving.
  It clearly terminates after a finite number of steps.

  In the following, we first argue that
  $|\mathcal H'| \geq \alpha|G|/2$ and then that
  $|G[\mathcal H']|\ge \alpha|G|/4d$.

  Let $v$ be a vertex that is chosen by the algorithm to remove all its incident copies.
  Once the vertex $v$ is chosen by the algorithm, it cannot be selected again.
  Since the degree of~$v$ at the moment of its selection is bounded by $\alpha\deg^G(v)/4d$, and as the copies in $\mathcal H$ are edge-disjoint,
  the number of copies removed in this step is at most $\alpha\deg^G(v)/4d$.

  We bound the total number of copies removed during the algorithm by summing the bounds over \emph{all} vertices $v$ of $G$:
  \[
    |\mathcal H|-|\mathcal H'|\le\sum_{v\in G}\frac{\alpha\deg^G(v)}{4d} = \frac{\alpha |E(G)|}{2d}.
  \]
  Since, furthermore, $|E(G)|\le d|G|$ and $|\mathcal H| \geq \alpha |G|$:
  \[
    |\mathcal H'| \geq |\mathcal H| - \frac{\alpha |E(G)|}{2d}
    \geq \alpha|G| - \frac{\alpha |E(G)|}{2d}
    \geq \frac{\alpha|E(G)|}{d} - \frac{\alpha |E(G)|}{2d}
    \geq \frac{\alpha|G|}{2}.
  \]
  Since $G[\mathcal H']$ is a subgraph of $G$,
  $G[\mathcal H']$ is also $d$-degenerate, and we have
  \[
    |G[\mathcal H']|\ge\frac{|E(G[\mathcal H'])|}{d}\ge\frac{\alpha|G|~|E(H)|}{2d}>\frac{\alpha|G|}{4d},
  \]
  since $|E(H)|\geq 1$. Therefore, $\mathcal H'$ is indeed $(\alpha/4d)$-degree preserving.
\end{proof}

\end{document}